\documentclass[11pt]{article}
\RequirePackage{algorithm,algorithmic,graphicx,amssymb,epsf,epic}

\def\denseformat{
\setlength{\textheight}{9.5in}
\setlength{\textwidth}{6.9in}
\setlength{\evensidemargin}{-0.3in}
\setlength{\oddsidemargin}{-0.3in}
\setlength{\headsep}{10pt}
\setlength{\topmargin}{-0.44in}
\setlength{\columnsep}{0.375in}
\setlength{\itemsep}{0pt}
}

\newtheorem{theorem}{Theorem}[section]

\newtheorem{lemma}[theorem]{Lemma}

\newtheorem{corollary}[theorem]{Corollary}

\newtheorem{comment}[theorem]{Comment}

\def\boldhead#1:{\par\vskip 7pt\noindent{\bf #1:}\hskip 10pt}
\def\ithead#1:{\par\vskip 7pt\noindent{\it #1:}\hskip 10pt}

\def\inline#1:{\par\vskip 7pt\noindent{\bf #1:}\hskip 10pt}
\def\midinline#1:{\par\noindent{\bf #1:}\hskip 10pt}
\def\dnsinline#1:{\par\vskip -7pt\noindent{\bf #1:}\hskip 10pt}
\def\ddnsinline#1:{\newline{\bf #1:}\hskip 10pt}
\def\largeinline#1:{\par\vskip 7pt\noindent{\large\bf #1:}\hskip 10pt}
\long\def\comment #1\commentend{}
\long\def\commhide #1\commhideend{}
\long\def\commfull #1\commend{#1}
\long\def\commabs #1\commenda{}
\long\def\commtim #1\commendt{#1}
\long\def\commb #1\commbend{}
\long\def\commedit #1\commeditend{} 

\long\def\commB #1\commBend{}       

\long\def\commex #1\commexend{}     

\long\def\commsiena #1\commsienaend{}  
                                         
\long\def\commBI #1\commBIend{}  
                                         
\long\def\CProof #1\CQED{}

\def\blackslug{\hbox{\hskip 1pt \vrule width 4pt height 8pt
    depth 1.5pt \hskip 1pt}}
\def\QED{\quad\blackslug\lower 8.5pt\null\par}

\long\def\PPP#1{\noindent{\bf Proof:}{ #1}{\quad\blackslug\lower 8.5pt\null}}

\long\def\denspar #1\densend
{#1}

\newif\ifnotesw\noteswtrue
   {\ifnotesw\marginpar[\hfill\(\top\)]{\(\top\)}\fi}%
   {\ifnotesw\marginpar[\hfill\(\bot\)]{\(\bot\)}\fi}

\newcommand{\mnote}[1]%
    {\ifnotesw\marginpar%
        [{\scriptsize\it\begin{minipage}[t]{\marginparwidth}
        \raggedleft#1%
                        \end{minipage}}]%
        {\scriptsize\it\begin{minipage}[t]{\marginparwidth}
        \raggedright#1%
                        \end{minipage}}%
    \fi}

\def\cA{{\cal A}}

\def\cE{{\cal E}}

\def\cH{{\cal H}}

\def\cL{{\cal L}}

\def\cQ{{\cal Q}}

\def\cU{{\cal U}}

\def\MathF{\hbox{\rm I\kern-2pt F}}
\def\MathP{\hbox{\rm I\kern-2pt P}}
\def\MathR{\hbox{\rm I\kern-2pt R}}
\def\MathZ{\hbox{\sf Z\kern-4pt Z}}
\def\MathN{\hbox{\rm I\kern-2pt I\kern-3.1pt N}}
\def\MathC{\hbox{\rm \kern0.7pt\raise0.8pt\hbox{\footnotesize I}
\kern-4.2pt C}}
\def\MathQ{\hbox{\rm I\kern-6pt Q}}

\newsavebox{\ttop}\newsavebox{\bbot}

\def\eps{\epsilon}

\denseformat

\newcommand {\ignore} [1] {}

\begin{document}

\title{\vspace{-0.02in} Fast Constructions of Light-Weight Spanners for General Graphs}
\author{
Michael Elkin \thanks{Department of Computer Science,
        Ben-Gurion University of the Negev, POB 653, Beer-Sheva 84105, Israel.
         \newline E-mail: {\tt \{elkinm,shayso\}@cs.bgu.ac.il}
          \newline Both authors are supported 
         by the BSF grant No.\ 2008430 and by the ISF grant No.\ 87209011.
         In addition, both authors are partially supported by the Lynn and William Frankel
        Center for Computer Sciences.}
\and
Shay Solomon $^*$\thanks{This research has been supported by the Clore Fellowship grant No.\ 81265410.}}

\date{\empty}

\begin{titlepage}
\def\thepage{}
\maketitle

\begin{abstract}
Since the pioneering works of Peleg and Sch$\ddot{\mbox{a}}$ffer \cite{PS89}, 
Alth$\ddot{\mbox{o}}$fer  et al.\ \cite{ADDJS93}, and Chandra et al.\ \cite{CDNS92},
it is known that for every \emph{weighted} undirected $n$-vertex $m$-edge graph $G = (V,E)$, 
and every integer $k \ge 1$, there exists a $(2k-1)$-spanner with $O(n^{1+1/{k}})$ edges
and weight $O(k \cdot n^{(1+\eps)/{k}}) \cdot \omega(MST(G))$, for any $\eps > 0$.
(Here $\omega(MST(G))$ stands for the weight of the minimum spanning tree of $G$.)
Nearly linear time algorithms for constructing $(2k-1)$-spanners with nearly $O(n^{1+{1}/{k}})$ edges were devised in \cite{BS03,RZ04,RTZ05}.
However, these algorithms fail to guarantee any meaningful upper bound on the weight of the constructed spanners.

To our knowledge, there are only two known algorithms for constructing sparse and light spanners for general graphs. 
One of them is the greedy algorithm of Alth$\ddot{\mbox{o}}$fer  et al.\ \cite{ADDJS93}, analyzed by Chandra et al.\ \cite{CDNS92}.  
The drawback of the greedy algorithm is that it requires $O(m \cdot (n^{1+{1}/{k}} +  n \cdot \log n))$ time.
The other algorithm is due to Awerbuch et al.\ \cite{ABP91}, from 1991. It constructs $O(k)$-spanners
with $O(k  \cdot n^{1+{1}/{k}} \cdot \Lambda)$ edges, weight $O(k^2 \cdot n^{{1}/{k}} \cdot \Lambda)
\cdot \omega(MST(G))$, within time $O(m \cdot k \cdot n^{{1}/{k}} \cdot \Lambda)$,
where $\Lambda$ is the logarithm of the aspect ratio of the graph.

The running time of both these algorithms is unsatisfactory. Moreover, the usually
faster algorithm of \cite{ABP91} pays for the speedup by significantly increasing
both the stretch, the sparsity, and the weight of the resulting spanner.

In this paper we devise an efficient algorithm for constructing sparse and light spanners.
Specifically, our algorithm constructs $((2k-1) \cdot (1+\eps))$-spanners with $O(k \cdot n^{1+{1}/{k}})$
edges and weight $O(k \cdot n^{{1}/{k}}) \cdot \omega(MST(G))$, where $\eps > 0$
is an arbitrarily small constant.
The running time of our algorithm is $O(k \cdot m  + \min\{n \cdot \log n,m \cdot \alpha(n)\})$.
Moreover, by slightly increasing the running time we can reduce the other parameters.
These results address an open problem from the ESA'04 paper by Roditty and Zwick \cite{RZ04}.
\end{abstract}
\end{titlepage}

\pagenumbering {arabic} 

\section{Introduction}
{\bf 1.1~ Centralized Algorithms.~}
Given an undirected weighted graph $G = (V,E)$ and a parameter $t \ge 1$, a subgraph $H = (V,E')$ of $G$
($E' \subseteq E$) is called a \emph{$t$-spanner} if for every edge $e = (u,v) \in E$, $dist_H(u,v) \le t \cdot dist_G(u,v)$.
(Here $dist_G(u,v)$ stands for the distance between $u$ and $v$ in the graph $G$.)
Graph spanners were introduced in 1989 by  Peleg and Sch$\ddot{\mbox{a}}$ffer \cite{PS89}
and Peleg and Ullman \cite{PU89}, who showed that for every unweighted $n$-vertex graph $G = (V,E)$ and an integer
parameter $k \ge 1$, there exists an $O(k)$-spanner with $O(n^{1+ 1/k})$ edges.
Alth$\ddot{\mbox{o}}$fer  et al.\ \cite{ADDJS93} improved and generalized these results. They analyzed
the natural greedy algorithm for constructing graph spanners, and showed  that for every $n$-vertex \emph{weighted}
graph $G = (V,E)$ and an integer parameter $k \ge 1$, this algorithm constructs a $(2k-1)$-spanner
with $O(n^{1+1/k})$ edges. (This is near-optimal \cite{PS89}.) They also showed that the weight of the resulting
spanner is $O(n/k) \cdot \omega(MST(G))$. (We will use the normalized notion of weight, called \emph{lightness}, which is the ratio
between the weight of the spanner and $\omega(MST(G))$.) In SoCG'92, Chandra et al.\ \cite{CDNS92} improved the lightness
bound, and showed that spanners obtained by the greedy algorithm have lightness 
$O(k \cdot n^{(1+\eps)/{k}})$, for an arbitrarily small $\eps > 0$.
However, the running time of their algorithm is $O(m \cdot (n^{1+{1}/{k}} +  n \cdot \log n))$, where $m$ stands for $|E|$.
Around the same time Awerbuch et al.\ \cite{ABP91} devised an algorithm that constructs $O(k)$-spanners with
$O(k  \cdot n^{1+{1}/{k}} \cdot \Lambda)$ edges  and lightness $O(k^2 \cdot n^{{1}/{k}} \cdot \Lambda)$, where $\Lambda$ is the logarithm of the aspect ratio of the input graph.
The running time of the algorithm of \cite{ABP91} is $O(m \cdot k \cdot n^{{1}/{k}} \cdot \Lambda)$.

In the two decades that passed since the results of \cite{PS89,PU89,ADDJS93,CDNS92,ABP91} graph spanners
turned out to be extremely useful. Among their applications is compact routing \cite{PU89,PU89b,TZ01},
distance oracles and labels \cite{Peleg99,TZ01b,RTZ05}, network synchronization \cite{Awerbuch85}, and computing almost shortest
paths \cite{Coh93,RZ04,Elkin05,EZ06,FKMSZ05}.
Graph spanners also became a subject of intensive research for their own sake \cite{Coh93,EP04,Elkin05,BS03,
TZ06,FKMSZ05,EZ06,Woodruff06,Pettie09,BKMP10}.

In particular, a lot of research attention was devoted to devising efficient algorithms for constructing sparse
spanners for weighted graphs. Cohen \cite{Coh93} devised a randomized algorithm for constructing $((2k-1)\cdot (1+\eps))$-spanners
with $O(k \cdot n^{1+1/k})$ edges. Her algorithm requires expected $O(m \cdot n^{1/k} \cdot k)$ time. Improving upon \cite{Coh93},
Baswana and Sen \cite{BS03} devised a randomized algorithm that constructs $(2k-1)$-spanners with
expected $O(k \cdot n^{1+1/k})$ edges, within expected $O(k \cdot m)$ time.
Roditty et al.\ \cite{RTZ05} derandomized this algorithm without any loss in parameters, or in running time.
Roditty and Zwick \cite{RZ04} devised a deterministic algorithm for constructing $(2k-1)$-spanners with
$O(n^{1+1/k})$ edges in $O(k \cdot n^{2+1/k})$ time. In the discussion section at the end of their
paper Roditty and Zwick \cite{RZ04} write: 

{\em ``Another interesting property of the (original) greedy algorithm, shown by 
\cite{CDNS92},
is that the total weight of the edges in the $(2k-1)$-spanner that it constructs is at
most $O(n^{(1+\eps)/k} \cdot \omega(MST(G)))$,\footnote{Actually, the weight in \cite{CDNS92} is $O(k \cdot n^{(1+\eps)/k} \cdot \omega(MST(G)))$.} for any $\eps > 0$.      
Unfortunately, this property no longer holds
for the modified greedy algorithm. Again, it is an interesting open problem to
obtain an efficient spanner construction algorithm that does have this property.''}

In the current paper we devise such a construction.
Specifically, our algorithm constructs $((2k-1) \cdot (1+\eps))$-spanners with $O(k \cdot n^{1+1/k})$ edges
and lightness $O(k \cdot n^{1/k})$, and does so in   time $O(k \cdot m + \min\{n \cdot \log  n, m \cdot \alpha(n)\})$,
where $\alpha(n)$ is an inverse Ackermann function. In other words, the running time of our algorithm is near-optimal and is drastically better
than the running time $O(m \cdot (n^{1+{1}/{k}} +  n \cdot \log n))$ 
of \cite{CDNS92} and than that of \cite{ABP91} ($O(m \cdot k \cdot n^{{1}/{k}} \cdot \Lambda)$).
We pay for this speed up by a small increase (by a factor of $(1+\eps)$) in the stretch and a small increase (by a factor of $k$) in the number of edges.
(This comparison is with \cite{CDNS92}. Our algorithm strictly outperforms the algorithm of \cite{ABP91}.)
We also have another variant  of our algorithm with a slightly higher running time ($O(k \cdot n^{2+1/k})$), and with
$O(n^{1+1/k})$ edges and lightness $O(k \cdot n^{1/k})$. 

Note that the relationship between the stretch and lightness in both our results
is essentially the same as in the state-of-the-art bound \cite{CDNS92}.
Specifically, in our result the slack factor $(1+\eps)$ appears in
the stretch, while in \cite{CDNS92} it appears in the exponent of the lightness.
The number of edges in our slower construction (that runs in $O(k \cdot n^{2+1/k})$ time)
is the same as in \cite{CDNS92}. The faster variant of our algorithm (that runs in near-optimal
time of $O(k \cdot m  + \min\{n \cdot \log n, m \cdot \alpha(n)\})$, 
pays for the speedup by increasing the number of edges by a factor of $k$.
See Table \ref{tab1} for a concise comparison
of our and previous results on light spanners. (The lightness of all other spanner constructions \cite{Coh93,BS03,RTZ05,RZ04}
is unbounded.)
\begin{table*}
\begin{center}
\resizebox {\textwidth }{!}{
\footnotesize
\begin{tabular}{|c||c|c|c|c|}
\hline results & stretch &  number of edges & lightness & running time  \\ 
\hline \hline  
~Alth$\ddot{\mbox{o}}$fer et al.\ \cite{ADDJS93}  & $2k-1$ & $O(n^{1+1/k})$  & $O(n/k)$ & $O(m \cdot n^{1+1/k})$  \\
\hline Chandra et al.\ \cite{CDNS92}  &  $2k-1$ & $O(n^{1+1/k})$ &  $O(k \cdot n^{(1+\eps)/k})$ & $O(m \cdot (n^{1+{1}/{k}} +  n \cdot \log n))$  \\
\hline Awerbuch et al.\ \cite{ABP91} & $O(k)$  & $O(k \cdot n^{1+1/k} \cdot \Lambda)$ & $O(k^2 \cdot n^{1/k} \cdot \Lambda)$ & $O(m \cdot k \cdot n^{{1}/{k}} \cdot \Lambda)$ \\
\hline \hline Our faster construction & {\boldmath $(2k-1)\cdot(1+\eps)$} & {\boldmath $O(k \cdot n^{1+1/k})$} & {\boldmath $O(k \cdot n^{1/k})$} & {\boldmath $O(k \cdot m + \min\{n \cdot \log n, m \cdot \alpha(n)\})$}  \\
\hline Our slower construction & {\boldmath $(2k-1)\cdot(1+\eps)$} & {\boldmath $O(n^{1+1/k})$} & {\boldmath $O(k \cdot n^{1/k})$} & {\boldmath $O(k \cdot n^{2+1/k})$}  \\
\hline
\end{tabular}
}
\end{center}
\caption[]{ \label{tab1} \footnotesize A  concise comparison of previous and our constructions of light spanners.
All the constructions mentioned in this table are deterministic. Our results are indicated by bold fonts.
}
\end{table*}

Chandra et al.\ \cite{CDNS92} also showed that the greedy algorithm gives rise to a construction  of $O(\log^2 n)$-spanners with $O(n)$
edges and constant lightness, and to a construction of $O(\log n)$-spanners with $O(n)$
edges and lightness $O(\log n)$. 
The running time of these constructions is $O(m \cdot n \cdot \log n)$.
Our algorithm also constructs spanners with the same (up to constant factors) parameters.
The running time required by our algorithm to construct these spanners is $O(n^2 \cdot \log n)$.
\vspace{0.1in}
\\{\bf 1.2~ Streaming Algorithms.~}
In the streaming model of computation the input graph $G = (V,E)$ arrives as a ``stream'', i.e.,
the algorithm reads edges one after another. The algorithm is required to process edges efficiently,
and to store only a limited amount of information. In the context of computing spanners the natural
memory limitation is the size of the spanner. Multi-pass streaming algorithms also allow several
(ideally, just a few) passes over the input stream.

The streaming model of computation was introduced by Alon et al.\ \cite{AMS99} and by
Feigenbaum et al.\ \cite{FKSV02}. The study of graph problems in the streaming model 
was introduced by Feigenbaum et al.\ \cite{FKMSZ05}. In particular, 
 Feigenbaum et al.\ \cite{FKMSZ05} devised a randomized one-pass streaming algorithm for
computing a $(2k+1)$-spanner with expected $O(k \cdot \log n \cdot n^{1+1/k})$ edges,
using $O(k \cdot \log n \cdot n^{1/k})$ processing time-per-edge.
This result was improved in \cite{Elkin11}, who devised a randomized one-pass streaming
algorithm that computes $(2k-1)$-spanners with expected $O(k \cdot n^{1+1/k})$ edges,
using $O(1)$ processing time-per-edge. See also \cite{Baswana08}.
Elkin and Zhang \cite{EZ06} devised a multi-pass streaming algorithm for constructing 
sparse $(1+\eps,\beta)$-spanners. The number of passes in their algorithm is $O(\beta)$.

To our knowledge, there are currently no efficient streaming algorithms for computing \emph{light}
spanners. We show that our algorithm can be implemented efficiently in the streaming model
\emph{augmented with the sorting primitives} (henceforth, \emph{augmented streaming model}).
This model, introduced by Aggarwal et al.\ \cite{ADRR04} in FOCS'94, allows to have ``sort passes''.
As a result of a sort pass, in consequent passes one can assume that the input stream that the algorithm
reads is sorted. (See \cite{ADRR04} for the justification of this model. The authors in \cite{ADRR04}
argue that ``streaming computations with an added sorting primitive are a natural and efficiently
implementable class of massive data set computations''.)

The algorithm of Chandra et al.\ \cite{CDNS92} can be viewed as an algorithm in this model.
After the initial sorting pass, it requires one pass over the input stream. As a result it constructs
a $(2k-1)$-spanner with $O(n^{1+1/k})$ edges and lightness $O(k \cdot n^{(1+\eps)/k})$, for
an arbitrarily small $\eps > 0$. The processing time-per-edge of this algorithm is, however, 
$O(n^{1+1/k})$, i.e., prohibitively large.

We show that a variant of our algorithm computes $((2k-1)\cdot(1+\eps))$-spanners with 
expected $O(k \cdot n^{1+1/k})$ edges and expected lightness $O(k^2 \cdot n^{1/k})$.
It performs two passes over the input stream, that follow an initial sorting pass.
In the first pass the worst-case (resp., amortized) processing time-per-edge of our
algorithm is $O(\frac{\log n}{\log \log n})$ (resp., $O(\alpha(n))$).
The processing time-per-edge of our algorithm in its second pass over the input stream
is $O(1)$.
\vspace{0.1in}
\\{\bf 1.3~ Our Techniques.~}
Our algorithm is based on a transformation, which given a black-box construction of sparse (possibly heavy) 
spanners with a certain stretch $t$, efficiently produces sparse and \emph{light} spanners with roughly
the same stretch. We use this transformation in conjunction with a number of known algorithms
that produce sparse spanners, but do not provide any bound on their lightness.

Our transformation generalizes a \emph{metric transformation} from \cite{CDNS92}.
Specifically, the metric transformation of \cite{CDNS92} converts constructions of sparse spanners for metrics
into constructions of sparse and light spanners (for the same metric). The
\emph{generalized transformation} that we devise applies to weighted not necessarily complete graphs.
(Observe that a metric can be viewed as a complete weighted graph.)

There are a number of technical difficulties that we overcome in our way to the generalized
transformation. Next, we briefly discuss one of them.
The construction of \cite{CDNS92} hierarchically partitions the point set of the input metric
into clusters. Then it selects a representative point from each cluster, and invokes its input
black-box construction of sparse spanners on the metric induced by the representatives.
One can try to mimic this approach in graphs by replacing each missing metric edge 
between representatives by the shortest path between them. This approach, however, is 
doomed to failure, as the overall number of edges taken into the spanner in this way
might be too large. To overcome this difficulty we carefully select \emph{representative edges}
which are inserted into a certain auxiliary graph. Then the black-box input construction is applied
to the auxiliary graph. 
As a result we obtain a spanner of the auxiliary graph, which we call \emph{auxiliary spanner}.
This auxiliary spanner $\cQ = (\cU,\cE)$ is a graph over a new vertex set $\cU$, i.e., $\cU$ is not a 
subset of the original vertex set $V$. Next, we ``project'' the auxiliary spanner $\cQ$ onto
the original graph, i.e., we translate edges of $\cE$ into edges of the original edge set $E$.
This needs to be done carefully, to avoid blowing up the stretch and lightness.
Also, it is crucial that this translation procedure will be efficient. Interestingly, we do not
project vertices of $\cU$ onto vertices of $V$, but rather edges of $\cE$ onto edges
of $E$. In particular, for a vertex $u \in \cU$ and two edges $(u,x),(u,y)\in \cE$, they may be
translated into two vertex-disjoint edges $(u',x'),(u'',y') \in E$.
As a result a path in $\cQ$ does not translate into a path in $G$, 
but rather into a collection of possibly vertex-disjoint edges. 
We show that these edges can be carefully glued into a path. 
This gluing, however, comes at a price of slightly increasing the stretch.
\vspace{0.1in}
\\{\bf 1.4~ Related Work.~} 
The large body  of work on constructing graph spanners efficiently was already discussed in Section 1.1.
The problem of constructing light spanners efficiently was also studied in the context of \emph{geometric spanners}.
See \cite{DN97} and \cite{GLN02}, and the references therein.
\vspace{0.1in}
\\{\bf 1.5~ Organization.~} 
In Section 2 we present and analyze our algorithm in the centralized model of computation. The algorithm is described in Section 2.1, and its analysis appears in Section 2.2.
In Section 3 we present a few variants of our basic algorithm (from Section 2.1). In particular, the streaming variant of our algorithm is presented in Section 3.5. 
\vspace{0.1in}
\\{\bf 1.6~ Preliminaries.~} 
We will use the following results as a black-box (we write $n = |V|, m = |E|$).
\begin{theorem} \cite{HZ96} [unweighted graphs] \label{unweighted}
For any unweighted graph $G = (V,E)$ and any integer $k \ge 1$,
a $(2k-1)$-spanner  with $O(n^{1+{1}/{k}})$ edges can be built in $O(m)$ time.
\end{theorem}

\begin{theorem} \cite{BS03,RTZ05} [weighted graphs I] \label{weighted}
For any weighted graph $G = (V,E)$ and any integer $k \ge 1$,
a $(2k-1)$-spanner with $O(k \cdot n^{1+{1}/{k}})$ edges can be built in $O(k \cdot m)$ time.
\end{theorem}
{\bf Remark:} The algorithm of \cite {BS03} is randomized, but was later derandomized in \cite{RTZ05}.
Henceforth, the algorithm provided by Theorem \ref{weighted} is deterministic.

\begin{theorem} \cite{RZ04} [weighted graphs II] \label{weighted2}
For any weighted graph $G = (V,E)$ and any integer $k \ge 1$,
a $(2k-1)$-spanner with $O(n^{1+{1}/{k}})$ edges can be built in $O(k \cdot n^{2+1/k})$ time.
\end{theorem}

\begin{theorem} \cite{Elkin11} [integer-weighted graphs] \label{weightedint}
For any integer-weighted graph $G = (V,E)$ and any integer $k \ge 1$,
a $(2k-1)$-spanner with \emph{expected} $O(k \cdot n^{1+{1}/{k}})$ edges can be built in $O(SORT(m))$ time,
where $SORT(m)$ is the time needed to sort $m$ integers.
\end{theorem}
{\bf Remark:} 
The algorithm of \cite{Elkin11} is randomized: while 
the guarantees $2k-1$ and $O(SORT(m))$ on the stretch and   running time, respectively, are deterministic,
the guarantee 
$O(k \cdot n^{1+{1}/{k}})$ on the size of the spanner is in expectation.
The fastest known randomized \cite{HT02} (respectively, deterministic \cite{Han04}) algorithm for sorting $m$
integers requires expected $O(m \cdot \sqrt{\log \log n})$ (resp., worst-case $O(m \cdot \log \log n)$) running time.

We will henceforth refer to the algorithms of Theorems \ref{unweighted}, \ref{weighted}, \ref{weighted2} and \ref{weightedint}
as Algorithms $UnwtdSp$, $WtdSp$, $WtdSp_2$, and $IntWtdSp$, respectively.

\section{The Basic Construction} \label{first}
\vspace{-0.06in}
\subsection{The Algorithm} \label{sec:alg}
\vspace{-0.06in}
In this section we devise an algorithm $LightSp$ that builds a light spanner efficiently.

Let $G = (V,E)$ be an arbitrary weighted graph, with $n = |V|, m = |E|$. Let $k \ge 1$ be an integer parameter
that determines the stretch bound of the spanner.


We start with building an MST (or an $O(1)$-approximate MST) $T = (V,E_T)$ for $G$.  
Let $M_T = (V,dist_T)$ be the (shortest-path) metric induced by $T$.   
We then compute the Hamiltonian path $\cL = (v_1,v_2,\ldots,v_n)$ of $M_T$ drawn by taking the preorder traversal of $T$.
Observe that for each $1 \le i \le n-1$, $dist_\cL(v_i,v_{i+1}) = dist_T(v_i,v_{i+1})$.    
More generally, for any pair $u,v \in V$ of vertices, $dist_\cL(u,v) \ge dist_T(u,v)$.  
Define $L = \omega(\cL)$; it is well known (\cite{CLRS90}, ch.\ 36) that $L \le 2 \cdot \omega(T)$, and so $L = O(\omega(MST(G)))$.

Let $1 < \rho \le 2$ be some parameter to be determined later, and define $\ell = \lceil \log_\rho n \rceil$.
We partition the edges of $E$ into $\ell + 1$ edge sets.
The first edge set $E_0$ contains all edges  with weight in the range $W_0 = (0, \frac{L}{n}]$.
For each $1 \le j \le \ell$,  
the $j$th edge set $E_j$ contains all edges with weight in the range
$W_j =  (\xi_j, \rho \cdot \xi_j]$, where $\xi_j = \rho^{j-1} \cdot \frac{L}{n}$.

Define $V_0 = V$, and let $n_0 = |V_0| = n$. 
We use Algorithm $WtdSp$ to build a $(2k-1)$-spanner $H'_0 = (V_0, E'_0)$ for $G_0 = (V,E_0)$.

Algorithm $LightSp$ proceeds in $\ell$ iterations $j = 1,2,\ldots,\ell$.
\begin{enumerate}
\item 
First, 
we divide the path $\cL$ into $n_j = \frac{q \cdot L}{\xi_j} = \frac{q\cdot n}{\rho^{j-1}}$
intervals of length $\mu_j = \frac{\xi_j}{q}$ each, where $\frac{1}{2k-1} < q < k$ is some parameter to be determined later.
(Notice that $n_j > n$, for all $1 \le j < \log_\rho q + 1$.)
These intervals induce a partition of   $V$ in the obvious\footnote{A vertex that ``lies'' on the boundary of two consecutive intervals can be assigned
to either one of them arbitrarily.} way; denote these intervals  
and the corresponding
vertex sets by 
$I^{(1)}_j,I^{(2)}_j,\ldots,I^{(n_j)}_{j}$
and $V^{(1)}_j,V^{(2)}_j,\ldots,V^{(n_j)}_{j}$, respectively. 
While computing this partition of $V$, we will store the index $i$ of the vertex set $V^{(i)}_j$ to which any vertex $v \in V$ 
belongs in some variable $ind_j(v)$, for each $1 \le i \le n_j$.
For each interval $I^{(i)}_j$, $1 \le i \le n_j$, we define
a new ``dummy'' vertex $r^{(i)}_j$ (i.e., not present in $G$) which will serve as
the \emph{representative} of the interval $I^{(i)}_j$. 
The vertices $r^{(i)}_j$ (respectively, intervals $I^{(i)}_j$), $1 \le i \le n_j$, will also be referred to as the \emph{$j$-level representatives}
(resp., \emph{$j$-level intervals}).
For any  vertex $v \in V$, 
its $j$-level representative is given by $r_j(v) = r^{(ind_j(v))}_j$.   
We also define a dummy vertex set $V_j$, which contains all the $j$-level representatives.
Observe that
$|V_j| ~\le~ \min\{n,n_j\} ~=~  \min\left\{n,\frac{q\cdot n}{\rho^{j-1}}\right\}$.
\item We then compute an edge set $\tilde E_j$ over $V_j$ as follows.
First, we remove from the $j$th edge set $E_j$ all edges that have both their endpoints in the same $j$-level interval
to obtain the edge set $\bar E_j$.
For each $e = (u,v) \in \bar E_j$,
let $\hat e = (r_j(u),r_j(v))$ be a  dummy edge of weight $\omega(e)$ between the corresponding $j$-level representatives,
and note that $r_j(u) \ne r_j(v)$;
we say that $\hat e = r(e)$ is the \emph{representative edge} of $e$, and  that $e = s(\hat e)$ is the \emph{source edge} of $\hat e$.
Define $\hat E_j = \{\hat e = r(e) ~|~ e \in \bar E_j\}$, and let $\hat G_j = (V_j,\hat E_j)$ be the corresponding multi-graph;
note that $\hat G_j$ does not contain self-loops.
Next, we transform $\hat G_j$ into a simple graph $\tilde G_j = (V_j, \tilde E_j)$ by removing all
edges but the one of minimum weight, for every pair of incident vertices in $\hat G_j$. 
During this process, we store with every edge $\tilde e \in \tilde E_j$ its source edge $e = s(\tilde e) \in \bar E_j$.
\item We proceed to building a $(2k-1)$-spanner $H'_j = (V_j, E'_j)$ for $\tilde G_j = (V_j, \tilde E_j)$
using Algorithm $WtdSp$.
\item Next, we replace each edge $e' \in E'_j$ by its source edge $s(e')$.
(Note that $e' \in E'_j \subseteq \tilde E_j$ is a dummy edge that connects some two distinct $j$-level intervals;
denote them by $I$ and $I'$. Moreover, its weight $\omega(e')$ is the smallest weight of an $E_j$-edge 
which connects $I$ and $I'$. Hence the source edge $s(e')$ is the minimum weight $E_j$-edge which
connects the intervals $I$ and $I'$.)
Denote the resulting edge set by $E^*_j$ and define $H^*_j = (V_j,E^*_j)$; note that $E^*_j \subseteq \bar E_j \subseteq E_j$.
We will refer to $H^*_j$ as the \emph{$j$-level spanner}.
\end{enumerate}

Finally, we define $E^* = E_T \cup E'_0 \cup \bigcup_{j=1}^\ell E^*_j$, and return the graph $H^* = (V,E^*)$ as our spanner.

\subsection{The Analysis} \label{sec:second}

In this section we analyze the properties of the constructed graph $H^*$. 
\vspace{0.1in}
\\{\bf 2.2.1~ Stretch.~}
We start with analyzing the stretch of the graph $H^* = (V,E^*)$. 


Define $str(k,q) = (2k-1)\cdot (1+\frac{2}{q}) + \frac{2}{q}$. 
Next, we show that the stretch of $H^*$ is
at most $str(k,q)$. 
\\In other words, we prove that $dist_{H^*}(u,v) \le  str(k,q) \cdot dist_G(u,v)$,
for any edge $e = (u,v) \in E$.


If $e \in E_0$, then we have $dist_{H^*}(u,v) ~\le~ dist_{H'_0}(u,v) ~\le~ (2k-1) \cdot dist_G(u,v)  ~\le~  str(k,q) \cdot dist_G(u,v).$

We henceforth assume that $e \in E_j$, for some   $1 \le j \le \ell$.
The next lemma shows that $dist_{H^*}(u,v) \le str(k,q) \cdot \omega(e)$.
(This lemma in conjunction with the triangle inequality
imply that $dist_{H^*}(u,v) \le str(k,q) \cdot dist_G(u,v)$,
for any edge $e = (u,v) \in E$, which provides the required stretch bound on $H^*$.)
\begin{lemma} \label{lm:stretch}
There is a path of length at most $str(k,q) \cdot \omega(e)$
in $H^*_j \cup T = (V,E^*_j \cup E_T)$
between  $u$ and $v$. 
\end{lemma}
\begin{proof}
Since $e$ is in $E_j$, its weight is in $W_j = [\xi_j,\rho \cdot \xi_j]$. In particular, we have $\omega(e) \ge \xi_j$.

For any pair $x,y \in V$ of vertices, let $\Pi_T(x,y)$ denote the path in the MST $T$ between them.

Suppose first that $u$ and $v$ belong to the same $j$-level interval. Note that the length of all $j$-level intervals
is $\mu_j = \frac{\xi_j}{q}$, and so $dist_T(u,v) \le dist_\cL(u,v) \le \frac{\xi_j}{q}$. Hence, the weight
of the path $\Pi_T(u,v)$ between $u$ and $v$ in $T$
is at most $$\frac{\xi_j}{q} ~\le~ \frac{\omega(e)}{q} ~<~ (2k-1) \cdot \omega(e) ~<~ str(k,q) \cdot \omega(e).$$
(The one before last inequality holds as $q > \frac{1}{2k-1}$.)
Thus $\Pi_T(u,v)$ is a path of the required length
in $H^*_j \cup T$ between $u$ and $v$.

We henceforth assume that $u$ and $v$ belong to distinct $j$-level intervals.
Let $\tilde e = (r_j(u),r_j(v)) \in \tilde E_j$ be the respective edge in $\tilde E_j$,
where $r_j(u)$ and $r_j(v)$ are the $j$-level representatives of $u$ and $v$, respectively.
By the construction of $\tilde E_j$, we have $\omega(\tilde e) \le \omega(e)$.
Let $\Pi'(r_j(u),r_j(v)) = (r_j(u) = u'_0,u'_1,\ldots,u'_h = r_j(v))$ be a path of length at most $(2k-1)\cdot dist_{\tilde G_j}(r_j(u),r_j(v))$ in $H'_j$
between $r_j(u)$ and $r_j(v)$. We have 
\begin{equation} \label{stretchspan}  
\omega(\Pi'(r_j(u),r_j(v))) ~\le~ (2k-1)\cdot dist_{\tilde G_j}(r_j(u),r_j(v)) ~\le~ (2k-1) \cdot \omega(\tilde e) ~\le~ (2k-1) \cdot \omega(e).
\end{equation}
Notice that the edges of $H'_j$ (and of $\Pi'(r_j(u),r_j(v))$)
are not taken as is to the graph $H^*_j$. 
That is, each edge $e'_i = (u'_i, u'_{i+1}) \in H'_j$ is replaced
by its source edge $s(e'_i) = (u_i,u_{i+1}) \in E^*_j$.
To connect the two endpoints  $u'_i$ and $u'_{i+1}$ of the original edge $e'_i \in E'_j$ in $H^*_j \cup T$,
we take the path  $\Pi(u'_i, u'_{i+1}) = \Pi_T(u'_i,u_i) \circ (u_i, u_{i+1}) \circ \Pi_T(u_{i+1},u'_{i+1})$
obtained from the concatenation of the path $\Pi_T(u'_i,u_i)$, the edge $s(e'_i) = (u_i, u_{i+1})$, and the path $\Pi_T(u_{i+1},u'_{i+1})$.
Observe that $u'_i$ is the $j$-level representative of $u_i$, i.e., $u'_i = r_j(u_i)$. In particular, they belong to the
same $j$-level interval, and so $\omega(\Pi_T(u'_i,u_i)) \le \frac{\xi_j}{q}$.
Similarly, $\omega(\Pi_T(u_{i+1},u'_{i+1})) \le \frac{\xi_j}{q}$.
Also, since $s(e'_i) \in E^*_j \subseteq E_j$, it holds that $\omega(e'_i) \ge \xi_j$.
Recall that the weight of the source edge $s(e'_i) = (u_i, u_{i+1})$ of $e'_i$ is equal to that of $e'_i = (u'_i,u'_{i+1})$,
and so $\omega(s(e'_i)) = \omega(e'_i) \ge \xi_j$.
Consequently, the weight $\omega(\Pi(u'_i,u'_{i+1}))$ of the path $\Pi(u'_i,u'_{i+1})$ satisfies
\begin{eqnarray}  \nonumber
\omega(\Pi(u'_i,u'_{i+1})) &=& \omega(\Pi_T(u'_i,u_i)) + \omega(u_i, u_{i+1}) + \omega(\Pi_T(u_{i+1},u'_{i+1}))
 \\ \label{stretchpath} &\le&   \frac{\xi_j}{q} + \omega(e'_i) +  \frac{\xi_j}{q} ~\le~ 
 \omega(e'_i) \cdot \left(1+ \frac{2}{q}\right).
\end{eqnarray}

Let $\Pi^*(r_j(u),r_j(v))$ be the path in $H^*_j \cup T$ between $r_j(u)$ and $r_j(v)$ obtained from $\Pi'(r_j(u),r_j(v))$
by replacing each edge $(u'_i,u'_{i+1})$ in it by the path $\Pi(u'_i,u'_{i+1}))$ as described above.
(See Figure \ref{stretchfig}.(I).)
Equations (\ref{stretchspan}) and (\ref{stretchpath}) yield $$\omega(\Pi^*(r_j(u),r_j(v))) ~\le~ \omega(\Pi'(r_j(u),r_j(v)))\cdot \left(1+\frac{2}{q}\right) ~\le~ (2k-1) \cdot \omega(e) \cdot \left(1+\frac{2}{q}\right).$$

\begin{figure*}[htp]
\begin{center}
\begin{minipage}{\textwidth} 
\begin{center}
\setlength{\epsfxsize}{7in} \epsfbox{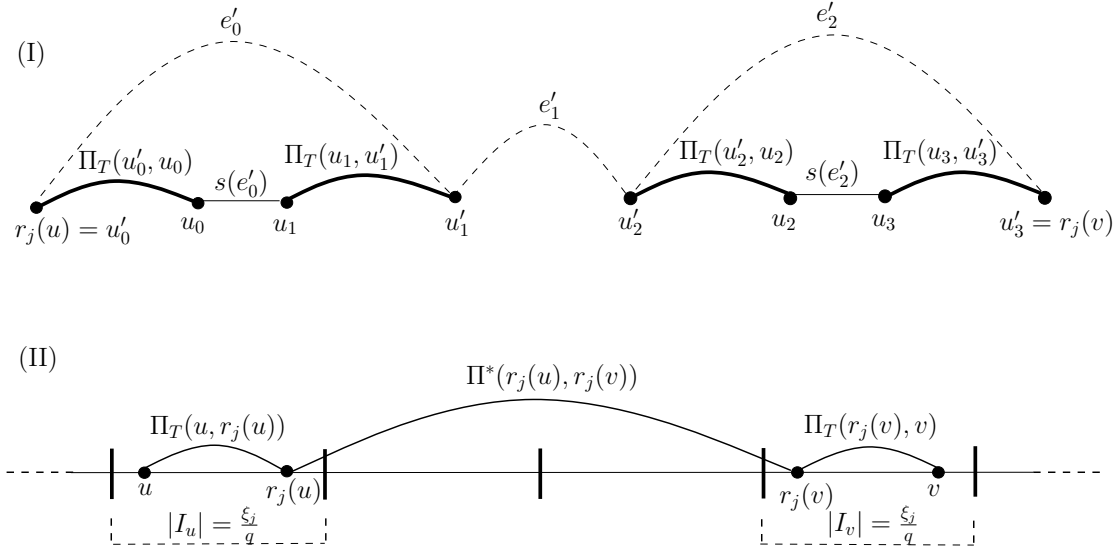}
\end{center}
\end{minipage}
\caption[]{ \label{stretchfig} \footnotesize \sf (I) An illustration of the path $\Pi^*(r_j(u),r_j(v))$. 
Edges $e'_i = (u'_i,u'_{i+1})$ of the path $\Pi'(r_j(u),r_j(v))$ in $H'_j$ are depicted by dashed curved lines.
The MST-paths $\Pi_T(u'_i,u_i) \in T$ between vertices $u'_i$ and $u_i$ that belong to the same interval are depicted by thick curved lines.
All the other edges (i.e., edges $(u_i, u_{i+1})$ in $H^*_j$) are depicted by thin straight lines.
(II) An illustration of the path $\Pi^*(u,v)$. 
The vertices $u$ and $r_j(u)$ (respectively, $v$ and $r_j(v)$) belong to the same
$j$-level interval, denoted here by $I_u$ (resp., $I_v$); the lengths of
these intervals (denoted here by $|I_u|$ and $|I_v|$)
are both equal to $\frac{\xi_j}{q}$.
}
\end{center}
\vspace{-0.1in}
\end{figure*}

Finally, let $\Pi^*(u,v) = \Pi_T(u,r_j(u)) \circ \Pi^*(r_j(u),r_j(v)) \circ \Pi_T(r_j(v),v)$
be the path in $H^*_j \cup T$ between $u$ and $v$ obtained from the concatenation of the paths $\Pi_T(u,r_j(u))$, $\Pi^*(r_j(u),r_j(v))$ and $\Pi_T(r_j(v),v)$.
(See Figure \ref{stretchfig}.(II) for an illustration.)
Since $u$ and $r_j(u)$ belong to the same $j$-level interval, we have $\omega(\Pi_T(u,r_j(u))) \le \frac{\xi_j}{q}$.
Similarly, $\omega(\Pi_T(r_j(v),v)) \le \frac{\xi_j}{q}$. Recall that $\omega(e) \ge \xi_j$.
We conclude that 
\begin{eqnarray*}
\omega(\Pi^*(u,v)) 
&=& \omega(\Pi_T(u,r_j(u))) + \omega(\Pi^*(r_j(u),r_j(v))) + \omega(\Pi_T(r_j(v),v))
\\ &\le&   \frac{\xi_j}{q} + (2k-1) \cdot \omega(e) \cdot \left(1+\frac{2}{q} \right) +  \frac{\xi_j}{q} ~\le~ \omega(e)\cdot \left((2k-1) \cdot \left(1+ \frac{2}{q}\right) + \frac{2}{q} \right) 
\\ &=& str(k,q) \cdot \omega(e).
\end{eqnarray*}
Hence $\Pi^*(u,v)$ is a path of length at most $str(k,q) \cdot \omega(e)$ between $u$ and $v$ in $H^*_j \cup T$, as required.
\QED
\end{proof}
\vspace{0.1in}
{\bf 2.2.2~ Number of Edges.~}
The next lemma bounds the number of edges in the spanner $H^* = (V,E^*)$.
\begin{lemma} \label{edgebound}
$|E^*| = O\left(k \cdot n^{1+{1}/{k}} \cdot  \left(1 +  \frac{q}{\rho-1}\right)\right)$.
\end{lemma}
\begin{proof}
Recall that $|V_j| ~\le~ \min\{n,n_j\} ~=~  \min\left\{n,\frac{q\cdot n}{\rho^{j-1}}\right\} ~\le~ n$, for $1 \le j \le \ell$.
For an index $\log_\rho q +1 \le j \le \ell$, we have $\frac{q}{\rho^{j-1}} \le 1$. Hence
$n \ge \frac{q \cdot n}{\rho^{j-1}} = n_j$, and so $|V_j| \le n_j$.
By construction, $E^* = E_T \cup E'_0 \cup \bigcup_{j=1}^\ell E^*_j$.  Clearly $|E_T| = n-1$. By Theorem \ref{weighted},
$|E'_0| = O(k \cdot |V_0|^{1+{1}/{k}}) = O(k \cdot n^{1+{1}/{k}})$.
Also, it holds that $|E^*_j| = |E'_j| = O(k \cdot |V_j|^{1+{1}/{k}}) = O(k \cdot n^{1+{1}/{k}})$, for each $1 \le j \le \ell$.
Next, consider an index $\log_\rho q +1 \le j \le \ell$; in this case  $|V_j| \le n_j = \frac{q \cdot n}{\rho^{j-1}}$.
Since $q < k$, we have $q^{1/k} = O(1)$.
It follows that \begin{equation} \label{exmark}
|E^*_j| ~=~ O\left(k \cdot \left(\frac{q \cdot n}{\rho^{j-1}}\right)^{1+{1}/{k}}\right) ~=~ 
O\left(k \cdot q \cdot n^{1+{1}/{k}} \cdot \left(\frac{1}{\rho^{1+{1}/{k}}}\right)^{j-1} \right).
\end{equation}
Since $\rho > 1$, we have
$$\sum_{\log_\rho q +1 \le j \le \ell} \left(\frac{1}{\rho^{1+{1}/{k}}}\right)^{j-1} ~\le~ \sum_{i=0}^\infty \left(\frac{1}{\rho^{1+{1}/{k}}}\right)^{i}
~\le~ \sum_{i=0}^\infty \left(\frac{1}{\rho}\right)^{i} ~=~ \frac{1}{\rho - 1} + 1.$$
It follows that
\begin{eqnarray*}
|E^*| &=& |E_T| + |E'_0| + \sum_{j=1}^\ell |E^*_j| ~=~ n - 1 + O(k \cdot n^{1+{1}/{k}})+ \sum_{1 \le j < \log_\rho q +1} |E^*_j| + \sum_{\log_\rho q +1 \le j \le \ell} |E^*_j|
\\ &=& O(k \cdot n^{1+{1}/{k}}) + O(\log_\rho q \cdot k \cdot n^{1+{1}/{k}}) + O\left(k \cdot q \cdot n^{1+{1}/{k}}\right) \cdot \sum_{\log_\rho q +1 \le j \le \ell}    \left(\frac{1}{\rho^{1+{1}/{k}}}\right)^{j-1} 
\\ &=&  O\left(k \cdot n^{1+{1}/{k}} \cdot  \left(1 +  \frac{q}{\rho-1}\right)\right). 
\end{eqnarray*}
(Note that for $q > 1$, it holds that $\log_\rho q = O(\frac{q}{\rho - 1})$.) \QED
\end{proof}
\vspace{0.1in}
{\bf 2.2.3~ Weight.~}
The next lemma estimates the weight of $H^* = (V,E^*)$.
\begin{lemma} \label{wtwt1}
$\omega(E^*) = O\left(k^2  \cdot n^{{1}/{k}}  \cdot \frac{q}{\rho - 1}\right) \cdot \omega(MST(G))$.
\end{lemma}
\begin{proof}
First, observe that the edge weights in $E_0$ are bounded above by $\frac{L}{n}$.
Since $|E'_0| =  O(k \cdot n^{1+{1}/{k}})$, 
it follows that 
\begin{equation} \label{realfirst}
\omega(E'_0) ~\le~ O(k \cdot n^{1+1/k}) \cdot \frac{L}{n} ~=~ O(k \cdot n^{1/k} \cdot L).
\end{equation}

The edge weights in $E_j$ are bounded above by $\xi_j \cdot \rho = \rho^{j} \cdot \frac{L}{n}$, for each $1 \le j \le \ell$.
\\Consider first indices $j$ in the range $1 \le j < \log_\rho q +1$. Since $|E^*_j| = O(k \cdot n^{1+1/k})$,
we have $\omega(E^*_j) = O(k \cdot n^{1+1/k}) \cdot (\rho^{j} \cdot \frac{L}{n})
=  O(k \cdot n^{1/k} \cdot \rho^{j} \cdot L)$.  
It follows that
\begin{equation} \label{sumfirst}
\sum_{1 \le j < \log_\rho q + 1} \omega(E^*_j) ~=~  O\left(k \cdot n^{1/k} \cdot L\right) \cdot
\sum_{1 \le j < \log_\rho q + 1} \rho^{j} ~=~ O\left(k \cdot n^{1/k} \cdot L \cdot q\right).
\end{equation}
Next, consider indices $j$ with $\log_\rho q +1 \le j \le \ell$. By (\ref{exmark}), $|E^*_j| = O\left(k \cdot q \cdot n^{1+1/k} \cdot \left(\frac{1}{\rho^{1+1/k}}\right)^{j-1} \right)$.
Hence $$\omega(E^*_j) ~=~ 
O\left(k \cdot q \cdot n^{1+1/k} \cdot \left(\frac{1}{\rho^{1+1/k}}\right)^{j-1} \right) \cdot \left(\rho^j \cdot \frac{L}{n}\right)
~=~
O\left(k \cdot q \cdot n^{1/k} \cdot \left(\frac{1}{\rho^{1/k}}\right)^{j-1} \cdot L \right).$$
Observe that $$\sum_{\log_\rho q + 1 \le j \le \ell} \left(\frac{1}{\rho^{1/k}}\right)^{j-1}
~\le~ \sum_{i=0}^\infty \left(\frac{1}{\rho^{1/k}}\right)^{i} ~=~ O\left(\frac{k}{\rho - 1}\right).$$
It follows that
\begin{eqnarray} \label{sumsecond}
\sum_{\log_\rho q + 1 \le j \le \ell} \omega(E^*_j) ~=~ 
O\left(k \cdot q \cdot n^{1/k} \cdot L \right) \cdot \sum_{\log_\rho q + 1 \le j \le \ell}  \left(\frac{1}{\rho^{1/k}}\right)^{j-1}
~=~ 
O\left(k^2 \cdot q \cdot n^{1/k} \cdot L  \cdot \frac{1}{\rho - 1}\right).
\end{eqnarray}
By Equations  (\ref{realfirst}), (\ref{sumfirst}) and (\ref{sumsecond}), 
\begin{eqnarray*}
\omega(E^*) &=& \omega(T) + \omega(E'_0) + \sum_{1 \le j < \log_\rho q + 1}\omega(E^*_j) + \sum_{\log_\rho q + 1 \le j \le \ell} \omega(E^*_j)
\\ &=& \omega(T) + O(k \cdot n^{1/k} \cdot L) + O\left(k \cdot n^{1/k} \cdot L \cdot q\right) + O\left(k^2 \cdot q \cdot n^{1/k} \cdot L  \cdot \frac{1}{\rho - 1}\right) \\ &=& 
O\left(k^2  \cdot n^{1/k} \cdot \frac{q}{\rho - 1}\right) \cdot \omega(MST(G)). 
\end{eqnarray*}
(Since $q > \frac{1}{2k-1}$, we have $k^2 \cdot q = \Omega(k)$.) \QED
\end{proof}
\vspace{0.1in}
{\bf 2.2.4~ Running Time.~}
The next lemma bounds the running time of our construction.
\begin{lemma} \label{lm:time}
The graph $H^*$ can be built within time $O(k \cdot m + \min\{n\cdot \log n, m \cdot \alpha(n)\})$. 
\end{lemma}
\begin{proof}
The tree $T$ can be built in time $O(\min\{m + n \cdot \log n, m \cdot \alpha(m,n)\})$ \cite{Chaz00}.
Computing the Hamiltonian path $\cL$ of $M_T$, as well as its weight $L = \omega(\cL)$, take another $O(n)$ time.
Having computed $L$, we can partition the edges of $E$ to the $\ell+1$ edge sets $E_0,E_1,\ldots,E_\ell$ in $O(m)$ time
in the obvious way.
\\By Theorem \ref{weighted}, building the spanner $H'_0$
for $G_0 = (V,E_0)$ requires $O(k \cdot |E_0|) = O(k \cdot m)$ time.

Next, we bound the running time of a single iteration $j$ in the main loop, for $1 \le j \le \ell$.
\begin{enumerate}
\item For each $1 \le j \le \ell$, we allocate an array $A_j$ of size $n_j$. For each
$x \in V_j$, the entry $A_j[x]$ will contain the linked list of $x$'s neighbors
in the multi-graph $\hat G_j$  (with the weights of the representative edges).
\\For each edge $e =(u,v) \in E$ we determine the index $j \in [\ell]$ such that
$\omega(e) \in W_j$ (and so $e \in E_j$). By the locations of $u$ and $v$ on the
path $\cL$, we determine the $j$-level representatives $r_j(u)$ and $r_j(v)$
of $u$ and $v$, respectively. 
We also determine whether the edge crosses between different $j$-level intervals.
If it is the case (i.e., $e \in \bar E_j$), then we insert the edge into the array $A_j$.
(In other words, we update the linked lists of both $r_j(u)$ and $r_j(v)$ in $A_j$.)
In this way we form the multi-graphs $\hat G_1,\hat G_2,\ldots,\hat G_\ell$ within
$O(|E|) = O(m)$ time. To prune these multi-graphs into simple graphs $\tilde G_1,\tilde G_2,\ldots,\tilde G_\ell$,
we need to remove all edges but the one of minimum weight, for every pair of incident vertices in $\hat G_j$.
This can be done deterministically within $O(m)$ time and space. 
(Alternatively, one can apply one of the algorithms for constructing $(2k-1)$-spanners with
$O(k \cdot n^{1+1/k})$ edges for weighted graphs \cite{BS03,RZ04,Elkin11} 
directly on the multi-graphs $\hat G_1,\hat G_2,\ldots,\hat G_\ell$, without pruning them first.
It is easy to verify that this does not affect their guarantees for stretch, size and running time.)
\item By Theorem \ref{weighted}, the time needed to build the $(2k-1)$-spanner $H'_j = (V_j,E'_j)$ for $\tilde G_j = (V,\tilde E_j)$
is $O(k \cdot |\tilde E_j|) = O(k \cdot |E_j|)$.
Summing over all $\ell$ iterations, the overall   time is $\sum_{j=1}^\ell O(k \cdot |E_j|) = O(k \cdot m)$. 
\item Replacing each edge of $E'_j$ by its source edge takes $O(1)$ time, thus the graph $H^*_j = (V,E^*_j)$
is built in  $O(|E'_j|) = O(|E_j|)$ time.
Summing over all $\ell$ iterations, the overall   time is $\sum_{j=1}^\ell O(|E_j|) = O(m)$. 
\end{enumerate} 

It follows that the total running time of the construction is \vspace{0.06in}\\$~~~~O(\min\{m + n \cdot \log n, m \cdot \alpha(m,n)\})
+ O(k \cdot m) ~=~ 
O(k \cdot m + \min\{n\cdot \log n, m \cdot \alpha(n)\})$. \QED
\end{proof}
We remark that there are randomized algorithms for constructing MST within $O(m+n)$ time \cite{KKT95}.
Employing one of them reduces the running time of our algorithm to just $O(k \cdot m)$.
\vspace{0.1in}
\\{\bf 2.2.5~ Summary.~}
We summarize the properties of the spanner $H^* = (V,E^*)$ in the following theorem.
(We substituted $\rho=2$ to optimize the parameters of the construction.)     
\begin{theorem} \label{th:weighted}
Let $G = (V,E)$ be a weighted graph, with $n = |V|, m = |E|$.
For any integer $k \ge 1$ and any number $\frac{1}{2k-1} < q < k$, a $\left((2k-1)\cdot (1+\frac{2}{q}) + \frac{2}{q}\right)$-spanner with 
 $O\left(k \cdot n^{1+{1}/{k}} \cdot  \left(1 +  q \right)\right)$   edges and lightness $O\left(k^2  \cdot n^{1/k} \cdot q \right)$,
 can be built in  $O(k \cdot m +  \min\{n\cdot \log n, m \cdot \alpha(n)\})$ time.
\end{theorem}

By substituting $q = \Theta(\frac{1}{\eps})$ in Theorem \ref{th:weighted},
for some small constant $\eps > 0$, we obtain:
\begin{corollary} \label{cor:try1}
Let $G = (V,E)$ be a  weighted graph, with $n = |V|, m = |E|$.
For any integer $k \ge 1$ and any constant $\eps > 0$, a $(\left(2k - 1)\cdot(1 +\eps)\right)$-spanner with 
$O\left(k \cdot n^{1+1/k} \right)$
   edges and lightness 
 $O\left(k^2 \cdot n^{1/k} \right)$,
 can be built in $O(k \cdot m +  \min\{n\cdot \log n, m \cdot \alpha(n)\})$ time. 
\end{corollary}

Note also that by substituting $q = \Theta(k/\eps)$, we get a $(2k-1+\eps)$-spanner with 
$O\left(k^2 \cdot n^{1+1/k} \right)$
   edges and lightness 
 $O\left(k^3 \cdot n^{1/k} \right)$, within the same time.

\section{Variants of the Construction}
In this section we devise some variants of the basic construction of Section \ref{first}
by applying a small (but significant) modification to Algorithm $LightSp$.

Notice that Algorithm $LightSp$ employs Algorithm $WtdSp$
as a black box for building (i) the spanner $H'_0 = (V_0,E'_0)$ for the graph $G_0 = (V,E_0)$,
and (ii) the spanner $H'_j = (V_j,E'_j)$ for the graph $\tilde G_j = (V_j,\tilde E_j)$ in step
3 of the main loop, for each $1 \le j \le \ell$. Instead of employing Algorithm $WtdSp$,
we can employ some of the other black-box spanners that are summarized in Section 1.6.

\subsection{First Variant (Slightly Increasing the Stretch)} \label{subfirst}
There is a significant difference between the spanner $H'_0$, and the spanners $H'_j$ for $j \ge 1$. While the aspect ratio (defined as the ratio between the maximum
and minimum edge weights) of the graph $G_0$ may be unbounded,
the aspect ratio of all the graphs $\tilde G_j, 1 \le j \le \ell$, is bounded above by $\rho$. This suggests
that we may view the graphs $\tilde G_j$ as unweighted, and will thus be able to use Algorithm $UnwtdSp$ for building the spanners $H'_j$ for them.
As a result, the stretch of each spanner $H'_j$ will increase by a factor of $\rho$. 

Denote by $\cH^*$ the variant of $H^*$ obtained by employing Algorithm $UnwtdSp$ for building the spanners $H'_j$, $1 \le j \le \ell$.
(We still use Algorithm $WtdSp$ to build the spanner $H'_0$.) 

We will next  analyze the properties of the resulting construction $\cH^*$.  
This analysis is  very similar to the analysis of the basic
construction that is given in Section 2.2, hence we aim for conciseness.
\vspace{0.1in}
\\{\bf Stretch.~}
We first show how to adapt the stretch analysis of Section 2.2.1 to the new construction $\cH^*$.
\\The proof of Lemma \ref{lm:stretch} carries through except for Equation (\ref{stretchspan}) that needs to be changed.
Observe that all weights of edges in $\tilde E_j$ belong to $W_j = [\xi_j,\rho \cdot \xi_j]$.
Hence they differ from each other by at most a multiplicative factor of $\rho$.
Note that $H'_j$ is an \emph{unweighted $(2k-1)$-spanner} for $\tilde G_j$, i.e.,  
a subgraph of $\tilde G_j$ that contains, for every edge $e \in \tilde E_j$,
a path with at most $(2k-1)$ edges between its endpoints.   
It is easy to see that $H'_j$ provides a $(\rho \cdot (2k-1))$-spanner for $\tilde G_j$. 
Hence, instead of Equation (\ref{stretchspan}), we will now have
$$\omega(\Pi'(r_j(u),r_j(v))) ~\le~ \rho \cdot (2k-1)\cdot dist_{\tilde G_j}(r_j(u),r_j(v)) ~\le~ \rho \cdot (2k-1) \cdot \omega(\tilde e) ~\le~ \rho \cdot (2k-1) \cdot \omega(e).$$
As a result, the upper bound on the weight of the path $\Pi^*(r_j(u),r_j(v))$ will also increase by a factor of $\rho$,
and we will have $\omega(\Pi^*(r_j(u),r_j(v))) \le \rho \cdot (2k-1) \cdot \omega(e) \cdot \left(1+\frac{2}{q}\right)$.
Consequently, we will get that $\omega(\Pi^*(u,v)) \le \omega(e)\cdot \left(\rho \cdot (2k-1) \cdot \left(1+ \frac{2}{q}\right) + \frac{2}{q} \right)$.
Hence the stretch of $\cH^*$ is $\rho \cdot (2k-1) \cdot (1+\frac{2}{q}) + \frac{2}{q}$. 
\vspace{0.1in}
\\{\bf Number of Edges.~}
Next, we show how to adapt the size analysis for the new construction $\cH^*$.
\\The only change from the proof of Lemma \ref{edgebound} is that now the upper bound
on $|E^*_j| = |E'_j|$, for $1 \le j \le \ell$, is smaller by a factor of $k$. The reason is
that we now use Algorithm $UnwtdSp$ rather than Algorithm $WtdSp$ to build the spanner $H'_j = (V_j, E'_j)$,  for $1 \le j \le \ell$;
compare Theorem \ref{unweighted} with Theorem \ref{weighted}.
Note that the bound on $|E'_0|$ remains unchanged, though, as we still use
Algorithm $WtdSp$  to build the spanner $H'_0$.
We will get
 \begin{eqnarray*}
|E^*| &=& |E_T| + |E'_0| + \sum_{j=1}^\ell |E^*_j| ~=~ n - 1 + O(k \cdot n^{1+{1}/{k}})+ \sum_{1 \le j < \log_\rho q +1} |E^*_j| + \sum_{\log_\rho q +1 \le j \le \ell} |E^*_j|
\\ &=& O(k \cdot n^{1+{1}/{k}})  + O(\log_\rho q  \cdot n^{1+{1}/{k}}) + O\left(q \cdot n^{1+{1}/{k}}\right) \cdot \sum_{\log_\rho q +1 \le j \le \ell}    \left(\frac{1}{\rho^{1+{1}/{k}}}\right)^{j-1} 
\\ &=& O\left(n^{1+{1}/{k}} \cdot \left(k +    \frac{q}{\rho-1}\right)\right). 
\end{eqnarray*}
{\bf Weight.~}
We now show how to adapt the weight analysis for the new construction $\cH^*$.
\\The only  change from the proof of Lemma \ref{wtwt1} is that now the upper bound
on $\omega(E^*_j)$, $1 \le j \le \ell$, is smaller by a factor of $k$.
The bound on $\omega(E'_0)$ remains unchanged (see Equation (\ref{realfirst})). 
Hence,
\begin{eqnarray*}
\omega(E^*) &=& \omega(T) + \omega(E'_0) + \sum_{1 \le j < \log_\rho q + 1}\omega(E^*_j) + \sum_{\log_\rho q + 1 \le j \le \ell} \omega(E^*_j)
\\ &=& \omega(T) + O(k \cdot n^{1/k} \cdot L) + O\left(n^{1/k} \cdot L \cdot q\right) + O\left(k \cdot q \cdot n^{1/k} \cdot L  \cdot \frac{1}{\rho - 1}\right) \\ &=& 
O\left(k \cdot n^{1/k}  \cdot \left(1 + \frac{q}{\rho - 1}\right) \right) \cdot \omega(MST(G)). 
\end{eqnarray*}
{\bf Running time.~} Clearly, the running time of $\cH^*$  is not higher than that of the basic construction.

We summarize the properties of the resulting construction $\cH^*$  in the following theorem.
\begin{theorem} \label{th:unweighted}
Let $G = (V,E)$ be a  weighted graph, with $n = |V|, m = |E|$.
For any integer $k \ge 1$, and any numbers $\frac{1}{2k-1} < q < k$ and $1 < \rho \le 2$, 
a $\left(\rho \cdot (2k-1) \cdot (1+\frac{2}{q}) + \frac{2}{q}\right)$-spanner with 
$O\left(n^{1+1/k} \cdot \left(k +   \frac{q}{\rho-1}\right)\right)$ 
   edges and lightness 
 $O\left(k \cdot n^{1/k}  \cdot \left(1 +  \frac{q}{\rho - 1}\right) \right)$,
 can be built in $O(k \cdot m +  \min\{n\cdot \log n, m \cdot \alpha(n)\})$ time. 
\end{theorem}
By substituting $q = \Theta(\frac{1}{\eps}), \rho = 1 + \Theta(\eps)$ in Theorem \ref{th:unweighted}, for an arbitrary constant $\eps > 0$, 
we obtain:
\begin{corollary} \label{cor:ltspspanner}
Let $G = (V,E)$ be a  weighted graph, with $n = |V|, m = |E|$.
For any integer $k \ge 1$ and any  small constant $\eps > 0$, a 
$\left((2k - 1) \cdot ( 1+\eps)\right)$-spanner 
with 
$O\left(k \cdot n^{1+1/k} \right)$
   edges and lightness 
 $O\left(k \cdot n^{1/k} \right)$,
 can be built in $O(k \cdot m +  \min\{n\cdot \log n, m \cdot \alpha(n)\})$ time. 
\end{corollary}
Corollary \ref{cor:ltspspanner}
in the particular case $k = O(\log n)$ 
gives rise to  an $O(\log n)$-spanner with $O(n \cdot \log n)$ edges and lightness $O(\log n)$.
The running time of this construction is $O(m \cdot \log n)$.

\subsection{Second Variant (Increasing the Running Time)}
Denote by $\check \cH$ the variant of $H^*$ obtained by employing Algorithm $UnwtdSp$ for building the spanners $H'_j$, $1 \le j \le \ell$,
and employing Algorithm $WtdSp_2$ (due to Roditty and Zwick \cite{RZ04}) to build the spanner $H'_0$.
It is easy to see that the stretch bound of the resulting construction $\check \cH$ is equal to that of $\cH^*$.
Also, the size and weight bounds of $\check \cH$ are better than those of $H^*$ and $\cH^*$, but the running time (which is dominated
by the running time of Algorithm $WtdSp_2$) is higher. 
The analysis of this variant is very similar to the analysis of the basic construction (in Section 2.2) and the analysis of the first variant (in Section \ref{subfirst}), and is thus omitted.

We summarize the properties of the resulting construction $\check \cH$ in the following theorem.
\begin{theorem} \label{th:unweightedver2}
Let $G = (V,E)$ be a weighted graph, with $n = |V|$.
For any integer $k \ge 1$, and any numbers $\frac{1}{2k-1} < q < k$ and $1 < \rho \le 2$, 
a $\left(\rho \cdot (2k-1) \cdot (1+\frac{2}{q}) + \frac{2}{q}\right)$-spanner with 
$O\left(n^{1+1/k} \cdot \left(1 +    \frac{q}{\rho-1}\right)\right)$ 
   edges and lightness
 $O\left(k \cdot n^{1/k}   \cdot \frac{q}{\rho - 1}\right)$,
 can be built in $O(k \cdot n^{2+1/k})$ time.
\end{theorem}
By substituting $q = \Theta(\frac{1}{\eps}),\rho = 1 + \Theta(\eps)$ in Theorem \ref{th:unweightedver2}, for an arbitrary constant $\eps > 0$, 
we obtain:
\begin{corollary} \label{highertime}
Let $G = (V,E)$ be a weighted graph, with $n = |V|$.
For any integer $k \ge 1$
and any small constant $\eps > 0$, a 
$\left((2k - 1) \cdot ( 1+\eps)\right)$-spanner 
with 
$O\left(n^{1+1/k} \right)$
   edges and lightness 
 $O\left(k \cdot n^{1/k} \right)$,
 can be built in $O(k \cdot n^{2+1/k})$ time. 
\end{corollary}
Note that the tradeoff between the stretch and lightness exhibited in Corollary \ref{highertime} is the same as in \cite{CDNS92}.
(There the stretch is $2k-1$, and the lightness is $O(k \cdot n^{(1+\eps)/k})$.) The number of edges is
$O(n^{1+1/k})$ in both results, but our running time is $O(k \cdot n^{2+1/k})$, instead of $O(m \cdot (n^{1+{1}/{k}} +  n \cdot \log n))$ in \cite{CDNS92}.

By substituting  $q = \Theta(1/k), \rho = 2$ in Theorem \ref{th:unweightedver2},
we obtain:
\begin{corollary}  \label{hope}
Let $G = (V,E)$ be a  weighted graph, with $n = |V|$.
For any integer $k \ge 1$, an $O(k^2)$-spanner with 
$O\left(n^{1+1/k} \right)$
   edges and lightness 
 $O\left(n^{1/k} \right)$,
 can be built in $O(k \cdot n^{2+1/k})$ time. 
\end{corollary}
Corollary \ref{hope} in the particular case $k = O(\log n)$ gives rise to an $O(\log^2 n)$-spanner
with $O(n)$ edges and lightness $O(1)$. We can extend this result
to get a general tradeoff between the three parameters
by substituting $k = O(\log n), \rho = 2$ in Theorem \ref{th:unweightedver2}, and taking $1 \le \ell = \frac{1}{q} = O(\log n)$.
\begin{corollary} \label{co:unweightedver2}
Let $G = (V,E)$ be a weighted graph, with $n = |V|$.
For any number $1 \le \ell = O(\log n)$, 
an $O(\log n \cdot \ell)$-spanner with 
$O(n)$ 
   edges and lightness
 $O(\frac{\log n}{\ell})$,
 can be built in $O(n^{2} \cdot \log n)$ time.
\end{corollary}
The tradeoff of Corollary \ref{co:unweightedver2} was also given by Chandra et al.\ \cite{CDNS92},
but their running time is $O(m \cdot  n \cdot \log n)$.

\subsection{Third Variant (Increasing the Running Time Some More)}
Denote by $\hat \cH$ the variant of $H^*$ obtained by employing Algorithm $WtdSp_2$ for building 
all the spanners $H'_j$, $1 \le j \le \ell$,
as well as the spanner $H'_0$.
It is easy to see that the stretch bound of the resulting construction $\hat \cH$ is equal to that of the basic construction $H^*$.
Also, the size and weight bounds of $\hat \cH$ are  exactly the same   as those of the second variant $\check \cH$,
but the running time is slightly higher (by a factor of $(1+q^2)$) than that of $\check \cH$. 
The analysis of this variant is very similar to the above,   
and is thus omitted.

We summarize the properties of the resulting construction $\hat \cH$ in the following theorem.
(We substituted $\rho=2$ to optimize the parameters of the construction.)     
\begin{theorem} \label{th:unweightedver3}
Let $G = (V,E)$ be a weighted graph, with $n = |V|$.
For any integer $k \ge 1$, and any number $\frac{1}{2k-1} < q < k$, 
a $\left((2k-1) \cdot (1+\frac{2}{q}) + \frac{2}{q}\right)$-spanner with 
$O\left(n^{1+1/k} \cdot \left(1 +  {q}\right)\right)$ 
   edges and lightness
 $O\left(k \cdot n^{1/k}   \cdot {q} \right)$,
 can be built in $O(k \cdot n^{2+1/k} \cdot (1+q^2))$ time.
\end{theorem}

\subsection{Fourth Variant (Integer-Weighted Graphs)}
The fourth variant of our construction applies to integer-weighted graphs only.
Denote by $\cH_{int}$ the variant of $H^*$ obtained by employing Algorithm $IntWtdSp$ for building 
all the black-box spanners, i.e., the spanners $H'_j$, $1 \le j \le \ell$,
and the spanner $H'_0$.
The new construction achieves exactly the same bounds as the basic construction $H^*$, except for the running time.
The running time of this variant  consists of two parts.
Specifically, it is the time required to compute the MST and the time required to compute
a spanner. On an integer-weighted graph the first  task can be done in $O(m+n)$ time
\cite{FW94}, while the second task requires $O(SORT(m))$ time (see Theorem 1.4 in \cite{Elkin11}).
Note, however, that the construction becomes randomized. 

We summarize the properties of the new construction $\cH_{int}$ in the following theorem.
\begin{theorem} \label{intweightedvar}
Let $G = (V,E)$ be an integer-weighted graph, with $n = |V|, m = |E|$.
For any integer $k \ge 1$ and any number $\frac{1}{2k-1} < q < k$,
a $\left((2k-1)\cdot (1+\frac{2}{q}) + \frac{2}{q}\right)$-spanner with 
expected $O\left(k \cdot n^{1+{1}/{k}} \cdot  \left(1 +  q \right)\right)$ 
   edges and expected lightness 
 $O\left(k^2  \cdot n^{1/k} \cdot q \right)$,
 can be built in $O(SORT(m))$ time.
\end{theorem}

\subsection{Fifth Variant (Spanners in the Streaming Model)}
In this section we analyze our algorithm in the augmented streaming model.
Specifically, this is the model introduced by Aggarwal et al.\ \cite{ADRR04}, which allows
sorting passes over the input. 

Our algorithm relies on a streaming algorithm
for constructing sparse (but possibly heavy) spanners from \cite{Elkin11},
summarized in Theorem \ref{th:elkin} below.
We remark that Theorem \ref{th:elkin} also applies to multi-graphs.
\begin{theorem} \cite{Elkin11}  \label{th:elkin}
For any unweighted $n$-vertex graph $G = (V,E)$ and any integer $k \ge 1$, there
exists a one-pass streaming algorithm that computes a $(2k-1)$-spanner
with $O(k \cdot n^{1+1/k})$ edges (expected).
The processing time-per-edge of the algorithm is $O(1)$, and its space requirement
is $O(k \cdot n^{1+1/k})$ (expected).
\\Also, in the augmented streaming model (if the algorithm accepts
a sorted stream of edges as input) the algorithm produces $(2k-1)$-spanners
with $O(k \cdot n^{1+1/k})$ edges (expected) for weighted graphs as well.
The processing time-per-edge and space requirements of this algorithm are
the same as in the unweighted case.
\end{theorem}

Our algorithm will run $\ell+1$ copies of the algorithm for weighted graphs from Theorem \ref{th:elkin}.
We will denote these copies $\cA_j$, for $0 \le j \le \ell$.

Our algorithm starts with a sorting pass over the stream of edges.
After this pass, in the consecutive pass the algorithm reads edges in a non-decreasing order
of weights. The objective of the first pass (after the sorting pass) is to compute an MST of
the input graph. To accomplish this, the algorithm maintains a Union-Find data structure
(see Ch.\ 21 in \cite{CLRS90}). It is known \cite{Blum86,ABR99} that all operations can
be performed in worst-case $O(\frac{\log n}{\log \log n})$ time, and in total $O(n + m \cdot \alpha(n))$
time, using $O(n)$ space.

As a result of the first pass, the MST $T$ is computed. The Hamiltonian path $\cL$ of $M_T$ is computed
between the passes. In addition, the algorithm maintains the location on $\cL$ of every vertex $v \in V$.
It also initializes $\ell+1$ arrays $A_j$ of size $n_j$ each, $0 \le j \le \ell$.

Then the algorithm performs the second pass over the sorted stream of edges.
For each edge $e = (u,v)$ that the algorithm reads in the second pass, the algorithm determines
the index $j$ such that $\omega(e) \in W_j$. Then it tests if the edge crosses between
different $j$-level intervals, i.e., belongs to $\bar E_j$. If it does not, this edge is skipped.
Otherwise the algorithm passes the edge $e$ to the $j$th copy $\cA_j$ of the
weighted streaming algorithm from Theorem \ref{th:elkin}.
The streaming algorithm $\cA_j$ will ultimately produce the spanner $H'_j$ for $\hat G_j = (V_j, \hat E_j)$.
(Here we follow the notation of Section \ref{sec:alg}.)
If the streaming algorithm $\cA_j$ decides to insert $e$ into the spanner $H'_j$, it will also 
insert its source edge $s(e)$ into the ultimate spanner $H^* = (V,E^*)$.
As a result the spanner $H^*$ will contain the union of the $j$-level spanners, 
for $0 \le j \le \ell$. The edge set of the MST will also be inserted into $H^*$ (either
between the passes, or after the second pass).

This completes the description of the algorithm. By Theorem \ref{th:elkin}, the (expected) space requirement in the second
pass is $\sum_{j=0}^\ell O(k \cdot n_j^{1+1/k}) = O(k \cdot n^{1+1/k})$.
The processing time-per-edge is $O(1)$. 

We summarize our streaming algorithm in the following theorem.
\begin{theorem}
In the augmented streaming model, for any weighted graph $G = (V,E)$,  
any integer $k \ge 1$, and any  constant $\eps > 0$, our algorithm requires two passes after
the sorting pass. It computes a $((2k-1)\cdot (1+\eps))$-spanner with expected $O(k \cdot n^{1+1/k})$
edges and expected lightness $O(k^2 \cdot n^{1/k})$.
The expected space requirement is $O(k \cdot n^{1+1/k})$.
The worst-case processing time-per-edge of the first (respectively, second) pass is $O(\frac{\log n}{\log \log n})$
(resp., $O(1)$). Moreover, the overall processing time of the first pass is $O(m \cdot \alpha(n))$.
\end{theorem}
As was mentioned in the introduction, the algorithm of \cite{CDNS92} can be viewed as an algorithm in this model.
It constructs $(2k-1)$-spanners with $O(n^{1+1/k})$ edges and lightness $O(k \cdot n^{(1+\eps)/k})$.
It requires one pass after the initial sorting pass, but its processing time-peg-edge is very large (specifically, it is $O(n^{1+1/k})$).
 \bibliographystyle{latex8}
\bibliography{latex8}

\end {document}